\documentclass[conference]{IEEEtran}

\usepackage{amsfonts}
\usepackage{amssymb,epic,eepic}
\usepackage{amsmath}
\usepackage{dcolumn}
\usepackage{bm}

\newcommand{\hr}{{\mathcal H}}
\newcommand{\cn}{{\mathcal N }}

\newcommand{\cs}{{\mathcal S}}
\newcommand{\crr}{{\mathcal R}}
\newcommand{\fr}{{\mathcal F}}

\newcommand{\fri}{{\mathfrak I}}
\newcommand{\kr}{{\mathcal K}}

\newcommand{\cc}{{\mathbb C}}
\newcommand{\rr}{{\mathbb R}}

\newcommand{\nn}{{\mathbb N}}
\newcommand{\idn}{\mathbf{1}}

\newcommand{\eps}{{\varepsilon}}        

\newcommand{\A}{\mathcal A}
\newcommand{\B}{\mathcal B}

\newcommand{\cP}{\mathcal P}
\newcommand{\bS}{\mathbf S}

\newcommand{\tr}{\mathrm{tr}}

\newtheorem{theorem}{Theorem}

\newtheorem{corollary}[theorem]{Corollary}

\newtheorem{definition}[theorem]{Definition}

\newtheorem{lemma}[theorem]{Lemma}

\begin{document}

\title{Entanglement Transmission over Arbitrarily Varying Quantum Channels}
\author{
\IEEEauthorblockN{Rudolf Ahlswede\IEEEauthorrefmark{1}, Igor Bjelakovi\'c\IEEEauthorrefmark{2}, Holger Boche\IEEEauthorrefmark{2} and Janis N\"otzel\IEEEauthorrefmark{2}}
\IEEEauthorblockA{\IEEEauthorrefmark{1}Working Group Information and Complexity, Universit\"at Bielefeld, Germany\\
Email: ahlswede@mathematik.uni-bielefeld.de}
\IEEEauthorblockA{\IEEEauthorrefmark{2}Heinrich Hertz-Lehrstuhl f\"ur Informationstheorie und theoretische Informationstechnik, Technische Universit\"at Berlin, Germany\\
Email: \{holger.boche, igor.bjelakovic, janis.noetzel\}@mk.tu-berlin.de}
}

\maketitle

\begin{abstract}
We derive a regularized formula for the common randomness assisted entanglement transmission capacity of finite arbitrarily varying quantum channels (AVQC's). For finite AVQC's with positive capacity for classical message transmission we show, by derandomization through classical forward communication, that the random capacity for entanglement transmission equals the deterministic capacity for entanglement transmission.\\
This is a quantum version of the famous Ahlswede dichotomy.\\
In the infinite case, we derive a similar result for certain classes of AVQC's. At last, we give two possible definitions of symmetrizability of an AVQC.
\end{abstract}
\section{\label{sec:introduction}Introduction}
We consider the task of entanglement transmission over an arbitrarily varying channel. This can be viewed as a three-party game in the following sense.\\
The sender's goal is to transmit one half of a maximally entangled state to the receiver by some (large) number of uses of a quantum channel which is under the control of a third party, called the adversary. The adversary is free to choose the channel out of a set of memoryless, partly nonstationary channels (cf. the beginning of section \ref{sec:codes-and-capacity}). Only this given set is previously known to both sender and receiver.\\
To make the situation even worse, the adversary knows the encoding-decoding procedure employed by sender and receiver, so that they have to choose this procedure such that it works well for all possible choices of channels that the adversary might come up with.\\
Earlier results in comparable situations have been obtained by Ahlswede \cite{ahlswede-elimination},\cite{ahlswede-coloring},\cite{ahlswede-gelfand-pinsker} for classical arbitrarily varying channels and Ahlswede and Blinovsky \cite{ahlswede-blinovsky} in the case of classical message transmission over an arbitrarily varying quantum channel.\\
In both cases we encounter a dichotomy: Either the capacity for classical message transmission over the arbitrarily varying (quantum) channel is zero or it equals its common-randomness assisted capacity. Also, for these models there exists the notion of \emph{symmetrizability}. This is a necessary and sufficient single-letter condition for an arbitrarily varying (quantum) channel to have zero capacity for message transmission.
Our work is based on ideas mainly taken from \cite{ahlswede-elimination}, \cite{ahlswede-coloring} and our earlier results for compound quantum channels \cite{bbn-2}.\\
The paper is organized as follows: In Section \ref{Notation and Conventions} we fix the basic notation. Section \ref{sec:codes-and-capacity} introduces our channel model, in Section \ref{sec:main-result} we summarize those of our results that lead to the quantum Ahlswede dichotomy. An outline of the strategy of proof is given in Section \ref{sec:outline-of-proof}. Finally, in Section \ref{sec:symmetrizability} we address the question of symmetrizability.\\
Details of the proofs given in this paper as well as the converse part of the coding theorem can be picked up in the accompanying paper \cite{abbn-1}.
\section{\label{Notation and Conventions}Notation and conventions}
All Hilbert spaces are assumed to have finite dimension and are over the field $\cc$. $\mathcal{S}(\hr)$ is the set of states, i.e. positive semi-definite operators with trace $1$ acting on the Hilbert space $\hr$. If $\fr\subset \hr$ is a subspace of $\hr$ then we write $\pi_{\fr}$ for the maximally mixed state on $\fr$, i.e. $\pi_{\fr}=\frac{p_{\fr}}{\tr(p_{\fr})}$ where $p_{\fr}$ stands for the projection onto $\fr$. For a finite set $A$, $\mathfrak{P}(A)$ denotes the set of probability distributions on $A$.\\
The set of completely positive trace preserving (CPTP) maps
between the operator spaces $\mathcal{B}(\hr)$ and
$\mathcal{B}(\kr)$ is denoted by $\mathcal{C}(\hr,\kr)$.
$\mathcal{C}^{\downarrow}(\hr,\kr)$ stands for the set of
completely positive trace decreasing maps between
$\mathcal{B}(\hr)$ and $\mathcal{B}(\kr)$.\\
We use the base two
logarithm which is denoted by $\log$. The von Neumann entropy of
a state $\rho\in\mathcal{S}(\hr)$ is given by
\[S(\rho):=-\textrm{tr}(\rho \log\rho).  \]
The coherent information for $\cn\in \mathcal{C}(\hr,\kr) $ and
$\rho\in\mathcal{S}(\hr)$ is defined by
\[I_c(\rho, \cn):=S(\cn (\rho))- S( (id_{\mathcal{B}(\hr)}\otimes \cn)(|\psi\rangle\langle \psi|)  ),  \]
where $\psi\in\hr\otimes \hr$ is an arbitrary purification of the state $\rho$. Following the usual conventions we let $S_e(\rho,\cn):=S( (id_{\mathcal{B}(\hr)}\otimes \cn)(|\psi\rangle\langle \psi|)  )$ denote the entropy exchange.\\
As a measure of entanglement preservation we use entanglement fidelity. For $\rho\in\mathcal{S}(\hr)$ and $\cn\in
\mathcal{C}^{\downarrow}(\hr,\kr)$ it is given by
\[F_e(\rho,\cn):=\langle\psi, (id_{\mathcal{B}(\hr)}\otimes \cn)(|\psi\rangle\langle \psi|)     \psi\rangle,  \]
with $\psi\in\hr\otimes \hr$ being an arbitrary purification of the state $\rho$.\\
We use the diamond norm $||\cdot||_\lozenge$ as a measure of closeness in the set of quantum channels, which is given by
\begin{equation}\label{def:diamond-norm}
||\cn||_{\lozenge}:=\sup_{n\in \nn}\max_{a\in \mathcal{B}(\cc^n\otimes\hr),||a||_1=1}||(id_{n}\otimes \mathcal{N})(a)||_1,   \end{equation}
where $id_n:\mathcal{B}(\cc^n)\to \mathcal{B}(\cc^n)$ is the identity channel, and $\mathcal{N}:\mathcal{B}(\hr)\to \mathcal{B}(\kr)$ is any linear map, not necessarily completely positive. The merits of $||\cdot||_{\lozenge}$ are due to the following facts (cf. \cite{kitaev}). First, $||\cn||_{\lozenge}=1$ for all $\cn\in\mathcal{C}(\hr,\kr)$. Thus, $\mathcal{C}(\hr,\kr)\subset S_{\lozenge}$, where $S_{\lozenge}$ denotes the unit sphere of the normed space $(\mathcal{B}(\mathcal{B}(\hr),\mathcal{B}(\kr)),||\cdot||_{\lozenge} )$. Moreover, $||\cn_1\otimes \cn_2||_{\lozenge}=||\cn_1||_{\lozenge}||\cn_2||_{\lozenge}$ for arbitrary linear maps $\cn_1,\cn_2:\mathcal{B}(\hr)\to \mathcal{B}(\kr) $. Finally, the supremum in (\ref{def:diamond-norm}) needs only be taken over $n$ that range over $\{1,2,\ldots,\dim\hr   \}.$\\
We further use the diamond norm to define the function $D_\lozenge(\cdot,\cdot)$ on $\{(\fri,\fri'):\fri,\fri'\subset\mathcal C(\hr,\kr)\}$, which is for $\fri,\fri'\subset\mathcal C(\hr,\kr)$ given by
\begin{eqnarray*}
&&D_\lozenge(\fri,\fri'):=\\
&&\max\{\sup_{\cn\in\fri}\inf_{\cn'\in\fri'}||\cn-\cn'||_\lozenge,\sup_{\cn'\in\fri'}\inf_{\cn\in\fri}||\cn-\cn'||_\lozenge\}.
\end{eqnarray*}
For compact sets, this is basically the Hausdorff distance induced by the diamond norm.\\
For arbitrary $\fri,\fri'\subset\mathcal C(\hr,\kr)$, $D_\lozenge(\fri,\fri')\leq\epsilon$ implies that for every $\cn\in\fri$ ($\cn'\in\fri'$) there exists $\cn'\in\fri'$ ($\cn\in\fri)$ such that $||\cn-\cn'||_\lozenge\leq2\epsilon$. In this way $D_\lozenge$ gives a measure of distance between sets of channels.\\
For an arbitrary set $\bS$, $\bS^l:=\{(s_1,\ldots,s_l):s_i\in\bS\ \forall i\in\{1,\ldots,l\}\}.$ We write $s^l$ for the elements of $\bS^l$.\\
For $\fri\subset\mathcal C(\hr,\kr)$ we denote its convex hull by $conv(\fri)$, a notation which is adapted from \cite{webster}.
\section{\label{sec:codes-and-capacity}Codes for entanglement and message transmission}
An arbitrarily varying quantum channel (AVQC) generated by a set $\fri=\{\cn_s  \}_{s\in \bS}$ of CPTP maps with input Hilbert space $\hr$ and output Hilbert space $\kr$ is the family of CPTP maps $\{\cn_{s^l}:\mathcal{B}(\hr)^{\otimes l}\to\mathcal{B}(\kr)^{\otimes l}  \}_{l\in\nn,s^l\in \bS^{l}}$, where
\[\cn_{s^l}:=\cn_{s_1}\otimes \ldots\otimes \cn_{s_l}\qquad \qquad (s^l\in\bS^l).  \]
In order to relieve ourselves of the burden of complicated notation we will simply write $\fri=\{\cn_s \}_{s\in\bS}$ for the AVQC.\\
Even in the case of a finite set $\fri=\{\cn_s  \}_{s\in\bS}$, showing the existence of reliable codes for the AVQC $\fri$ is a non-trivial task, since for each block length $l\in\nn$ we have to deal with $|\fri|^{l}$ memoryless partly non-stationary quantum channels simultaneously.\\
Let $\fri=\{\cn_s\}_{s\in\bS}$ be an AVQC.\\
An $(l,k_l)-$\emph{random entanglement transmission code} for $\fri$ is a probability measure $\mu_l$ on $(\mathcal C(\fr_l,\hr)\times\mathcal C(\kr,\fr_l'),\sigma_l)$, where $\dim\fr_l=k_l$, $\fr_l\subset\fr_l'$ and the sigma-algebra $\sigma_l$ is chosen such that $F_e(\pi_{\fr_l},(\cdot)\circ\cn_{s^l}\circ(\cdot))$ is measurable w.r.t. $\sigma_l$ for every $s^l\in\bS^l$. Moreover, we assume that $\sigma_l$ contains all singleton sets. An example of such a sigma-algebra $\sigma_l$ is given by the product of sigma-algebras of Borel sets induced on $\mathcal C(\fr_l,\hr) $ and $\mathcal C(\kr,\fr_l') $ by the standard topologies of the ambient spaces.
\begin{definition}\label{def:random-cap-ent-trans}
A non-negative number $R$ is said to be an achievable entanglement transmission rate for $\fri$ with random codes if there is a sequence of $(l,k_l)-$random entanglement transmission codes such that\\
1) $\liminf_{l\rightarrow\infty}\frac{1}{l}\log k_l\geq R$ and\\
2) $\lim_{l\rightarrow\infty}\inf_{s^l\in\bS^l}\int F_e(\pi_{\fr_l},\crr^l\circ\cn_{s^l}\circ\cP^l)d\mu_l(\cP^l,\crr^l)=1$.\\
The random capacity $\A_{\textup{r}}(\fri)$ for entanglement transmission over $\fri$ is defined by
\begin{eqnarray*}
\A_{\textup{r}}(\fri):=\sup\{R&:&R \textrm{ is an achievable entanglement trans-}\\
&&\textrm{mission rate for } \fri \textrm{ with random codes}\}.
\end{eqnarray*}
\end{definition}
We are now in a position to introduce deterministic codes: An $(l,k_l)-$code for entanglement transmission over $\fri$ is an $(l,k_l)-$random code for $\fri$ with $\mu_l(\{(\mathcal{P}^l,\crr^l)  \}  )=1$ for some encoder-decoder pair $(\mathcal{P}^l,\crr^l)$\footnote{This explains our requirement on $\sigma_l$ to contain all singleton sets.} and $\mu_l(A)=0$ for any $A\in\sigma_l$ with $(\mathcal{P}^l,\crr^l)\notin A $. We will refer to such measures as point measures in what follows.
\begin{definition}
A non-negative number $R$ is a deterministically achievable rate for entanglement transmission over $\fri$ if it is achievable in the sense of Definition \ref{def:random-cap-ent-trans} for random codes  with \emph{point measures} $\mu_l$.\\
The deterministic capacity $\A_{\textup{d}}(\fri)$ for entanglement transmission over the AVQC $\fri$ is  given by
\begin{eqnarray*}
\A_{\textup{d}}(\fri):=\sup \{R&:& R \textrm{ is a deterministically achievable rate}\\
&&\textrm{for entanglement transmission over }\fri  \}.
\end{eqnarray*}
\end{definition}
Finally, we shall need the notion of the classical deterministic capacity $C_{\textrm{det}}(\fri)$ of the AVQC $\fri=\{\cn_s  \}_{s\in\bS}$ with \emph{average} error criterion.
An $(l,M_l)$-(deterministic) code for message transmission is a family of pairs $\mathfrak{C}_l=(\rho_i, D_i)_{i=1}^{M_l}$
where $\rho_1,\ldots ,\rho_{M_l}\in\cs(\hr^{\otimes l})$, and positive semi-definite operators $D_1,\ldots, D_{M_l}\in\mathcal{B}(\kr^{\otimes l})$ satisfying $\sum_{i=1}^{M_l}D_i=\idn_{\kr^{\otimes l}} $. The underlying error criterion we shall use is the worst-case average probability of error of the code $\mathfrak{C}_l$ which is given by
\begin{equation}\label{def-av-error}
\bar P_{e,l}(\fri):=\sup_{s^l\in \bS^l}\bar P_e (\mathfrak{C}_l,s^l),
\end{equation}
where for $s^l\in\bS^l$ we set
\[ P_e (\mathfrak{C}_l,s^l):=\frac{1}{M_l}\sum_{i=1}^{M_l}\left(1- \tr(\cn_{s^l}(\rho_i)D_i)  \right).  \]
The achievable rates and the classical deterministic capacity  $C_{\textrm{det}}(\fri)$ of $\fri$, with respect to the error criterion given in (\ref{def-av-error}), are then defined in the usual way (see e.g. \cite{ahlswede-blinovsky}).\\
\section{\label{sec:main-result}Main results}
The compound quantum channel generated by $conv(\fri)$ (cf. \cite{bbn-2} for the relevant definitions) shall play the crucial role in our derivation of the coding results stated below.\\
Our main result, a quantum version of Ahlswede's dichotomy for finite AVQCs, goes as follows:
\begin{theorem}\label{quant-ahlswede-dichotomy}
Let $\fri=\{\cn_s  \}_{s\in \bS}$ be a finite AVQC.
\begin{enumerate}
\item The random capacity for entanglement transmission over $\fri$ is given by
  \begin{equation}\label{eq:ahlswede-dichotomy-1}
    \A_{\textup{r}}(\fri)=\lim_{l\to\infty}\frac{1}{l}\max_{\rho\in\cs(\hr^{\otimes l})}\inf_{\cn\in conv(\fri)}I_c(\rho, \cn^{\otimes l}).
  \end{equation}
\item Either $C_{\textup{det}}(\fri)=0 $ or else $\A_{\textup{d}}(\fri)= \A_{\textup{r}}(\fri)$.
\end{enumerate}
\end{theorem}
\emph{\underline{Remark}. It is clear from convexity of entanglement fidelity in the input state that $\A_{\textup{d}}(\fri)\le C_{\textup{det}}(\fri) $, so that $C_{\textup{det}}(\fri)=0 $ implies  $\A_{\textup{d}}(\fri)=0 $. Therefore, Theorem \ref{quant-ahlswede-dichotomy} determines  $\A_{\textup{d}}(\fri) $, in principle, up to required regularization on the right-hand side of (\ref{eq:ahlswede-dichotomy-1}) and the question of when $C_{\textup{det}}(\fri)=0$ happens. We derive a non-single-letter necessary and sufficient condition for the latter in Section \ref{sec:symmetrizability}.
}\\
In the case that $\bS$ is infinite, we have the following statement:
\begin{theorem}\label{avqc-infinite}
Let $\fri=\{\cn_s  \}_{s\in\bS}$ be any AVQC and $\partial\mathcal C$ the topological boundary of $\mathcal C(\hr,\kr)$. If $D_{\lozenge}(\tilde\fri,\partial \mathcal{C})>0$, then
\[\A_{\textup{r}}(\fri)=\lim_{l\to\infty}\frac{1}{l}\max_{\rho\in\cs(\hr^{\otimes l})}\inf_{\cn\in conv(\fri)}I_c(\rho,\cn^{\otimes l}).  \]
\end{theorem}
\emph{\underline{Remark}.
The condition $D_{\lozenge}(\tilde\fri,\partial \mathcal{C})>0$ in Theorem \ref{avqc-infinite} stems from our strategy of approximation of an infinite AVQC through a sequence of finite AVQC's. We hope to be able to drop this artificial constraint in the final version of the paper.
}
\section{\label{sec:outline-of-proof}Outline of the proof}
This section is split into three parts. First, we demonstrate the existence of asymptotically optimal sequences of random codes (in the sense of (\ref{eq:ahlswede-dichotomy-1})). We use Ahlswede's robustification technique originally presented in \cite{ahlswede-coloring} in the form presented in \cite{ahlswede-gelfand-pinsker} and our results on compound quantum channels \cite{bbn-2} in order to get a sequence of finitely supported probability measures $\mu_l$ on the set of encoding and decoding maps. Second, we show that the support of each $\mu_l$ can be taken as a set with cardinality $l^2$.\\
Third, we show that $C_d(\fri)>0$ implies that we can derandomize our code without any asymptotic loss of capacity, so that $\A_d(\fri)=\A_r(\fri)$ holds.\\
Fourth, we briefly sketch how approximation of $conv(\fri)$ by convex polytopes leads to Theorem \ref{avqc-infinite}.
\subsection{\label{subsec:random-achieve-finite}Finite AVQC}
Let $l\in\nn$ and let $\mathbf{P}_l$ denote the set of permutations acting on $\{1,\ldots, l\}$. Suppose we are given a finite set $\bS$. Then each permutation $P\in \mathbf{P}_l$ induces an action on $\bS^l$ by  $P:\mathbf S^l\rightarrow\mathbf S^l$, $P(s^l)_i:=s_{P(i)}$. By $T(l,\bS)$, we denote the set of types on $\bS$ induced by the elements of $\bS^l$, i.e. the set of empirical distributions on $\bS$ generated by sequences in $\bS^l$. Now Ahlswede's robustification can be stated as follows.
\begin{theorem}[Robustification technique, cf. \cite{ahlswede-gelfand-pinsker}]\label{robustification-technique}
If a function $f:\bS^l\to [0,1]$ satisfies
\begin{equation}\label{eq:robustification-1}
 \sum_{s^l\in\bS^l}f(s^l)q(s_1)\cdot\ldots\cdot q(s_l)\ge 1-\gamma
\end{equation}
for all $q\in T(l,\bS)$ and some $\gamma\in [0,1]$, then
\begin{equation}\label{eq:robustification-2}
  \frac{1}{l!}\sum_{P\in\mathbf{P}_l}f(P(s^l))\ge 1-(l+1)^{|\bS  |}\cdot \gamma\qquad \forall s^l\in \bS^l.
\end{equation}
\end{theorem}\ \\
As another ingredient for the arguments to follow we need an achievability result for the compound channel $conv(\fri)$. We set for $k\in\nn$
\[conv(\fri)^{\otimes k}:=\{\cn_q^{\otimes k}  \}_{q\in\mathfrak{P}(\bS)}. \]
\begin{lemma}\label{compound-achiev-input}
Let $k\in \nn$. Suppose that
\[\max_{\rho\in\cs(\hr^{\otimes k})}\inf_{\cn\in conv(\fri)^{\otimes k}}I_c(\rho,\cn)>0  \]
holds. Then for each sufficiently small $\eta>0$ there is a sequence of $(l,k_l)$-codes $(\mathcal{P}^l, \crr^l)_{l\in\nn}$ such that for all $l\ge l_0(\eta)$ the inequalities
\begin{equation}\label{eq:comp-input-1}
  F_{e}(\pi_{\fr_l}, \crr^l\circ \cn^{\otimes l}\circ \mathcal{P}^l)\ge 1- 2^{-lc}\qquad \forall \cn\in conv(\fri),
\end{equation}
\begin{equation}
  \frac{1}{l}\log \dim \fr_l\ge \frac{1}{k}\max_{\rho\in\cs(\hr^{\otimes k})}\inf_{\cn\in conv(\fri)^{\otimes k}}I_c(\rho,\cn)  -\eta,
\end{equation}
hold with a constant $c=c(k,\dim \hr,\dim\kr, conv(\fri),\eta)>0$.
\end{lemma}
\begin{proof}The proof follows from an application of the compound BSST Lemma and Lemma 9 in \cite{bbn-2}. These two statements show the existence of well behaved codes for the channels $\cn_q^{\otimes m\cdot k}$, where $m$ depends on $conv(\fri),\ k$ and $\eta$. For fixed $k$, all we have to do is convert these codes to codes for the channels $\cn_q$.
\end{proof}
In the next step we will combine the robustification technique and Lemma \ref{compound-achiev-input} to prove the existence of good random codes for the AVQC $\fri=\{\cn_s  \}_{s\in\bS}$.\\
Recall that there is a canonical action of $\mathbf{P}_l$ on $\mathcal{B}(\hr)^{\otimes l}$ given by $P_\hr(a_1\otimes\ldots\otimes a_l):=a_{P^{-1}(1)}\otimes\ldots\otimes a_{P^{-1}(n)}$. It is easy to see that $P_\hr(a)=U_{P}aU_{P}^{\ast},\ (a\in\B(\hr)^{\otimes l})$ with the unitary operator $U_{P}:\hr^{\otimes l}\to\hr^{\otimes l}$ defined by $U_{P}(x_1\otimes \ldots\otimes x_l)=x_{P^{-1}(1)}\otimes \ldots\otimes x_{P^{-1}(l)}$.
\begin{theorem}[Conversion of compound codes]\label{conversion-of-compound-codes}
Let $\fri=\{\cn_s  \}_{s\in\bS}$ be a finite AVQC. For each $k\in\nn$ and any sufficiently small $\eta>0$ there is a sequence of $(l,k_l)$-codes $(\mathcal{P}^l,\crr^l)_{l\in\mathbb N}$, $\mathcal{P}^l\in\mathcal{C}(\fr_l,\hr^{\otimes l}),\crr^l\in\mathcal{C}(\kr^{\otimes l},\fr'_l)$, for the compound channel built up from $conv(\fri)$ satisfying
\begin{equation}\label{conversion-1}
\frac{1}{l}\log \dim \fr_l\ge \frac{1}{k}\max_{\rho\in\cs(\hr^{\otimes k})}\inf_{\cn\in conv(\fri)^{\otimes k}}I_c(\rho,\cn)  -\eta
\end{equation}
and, for all sufficiently large $l\in\nn$ and $s^l\in\bS^l$,
\begin{equation}\label{conversion-2}
\sum_{P\in\mathbf{P}_l}\frac{1}{l!}F_e(\pi_{\fr_l}, \crr^l\circ P_\kr^{-1}\circ \cn_{s^l}\circ P_\hr\circ \mathcal{P}^l)\ge 1- (l+1)^{|\bS|}\cdot 2^{-lc},
\end{equation}
with a positive number $c=c(k,\dim\hr,\dim\kr,conv(\fri),\eta)$.
\end{theorem}
\begin{proof}We let $(\crr^l,\mathcal P^l)$ be as in Theorem \ref{compound-achiev-input}. Setting $f(s^l):= F_{e}(\pi_{\fr_l}, \crr^l\circ \cn_{s^l}\circ \mathcal{P}^l)$ and applying Theorem \ref{robustification-technique} proves the theorem.
\end{proof}
For $l\in\nn$, define a discretely supported probability measure $\mu_l$ by
\[\mu_l:=\frac{1}{l!}\sum_{P\in\mathbf{P}_l}\delta_{(P_\hr\circ \mathcal{P}^l, \crr^l\circ P^{-1}_\kr)}, \]
where $\delta_{(P_\hr\circ \mathcal{P}^l, \crr^l\circ P^{-1}_\kr)} $ denotes the probability measure that puts measure $1$ on the point $(P_\hr\circ \mathcal{P}^l, \crr^l\circ P^{-1}_\kr) $,
we obtain for each $k\in\nn$ a sequence of $(l,k_l)$-random codes achieving
\[ \frac{1}{k}\max_{\rho\in\cs(\hr^{\otimes k})}\inf_{\cn\in conv(\fri)^{\otimes l}}I_c(\rho,\cn). \]
This leads to the following corollary to Theorem \ref{conversion-of-compound-codes}.
\begin{corollary}\label{achievability-finite-avqc}
For any finite AVQC $\fri=\{\cn_s  \}_{s\in\bS}$ we have
\[\mathcal{A}_{\textup{r}}(\fri)\ge \lim_{l\to\infty}\frac{1}{l}\max_{\rho\in\cs(\hr^{\otimes l})}\inf_{\cn\in conv(\fri)}I_c(\rho,\cn^{\otimes l}).  \]
\end{corollary}
\subsection{\label{subsec:derandomization}Derandomization}

In this section we will prove the second claim made in Theorem \ref{quant-ahlswede-dichotomy} by following Ahlswede's elimination technique. The proof is based on the following lemma, which shows that not much of common randomness is needed to achieve $\A_{\textrm{r}}(\fri)$.
\begin{lemma}[Random Code Reduction]\label{random-code-reduction}
Let $\fri=\{\cn_s\}_{s\in\mathbf S}$ be a finite AVQC, $l\in\nn$, and $\mu_l$ an $(l,k_l)$-random code for the AVQC $\fri$ with
\begin{equation}\label{eq:random-code-reduction}
\min_{s^l\in\bS^l}\int F_e(\pi_{\fr_l},\crr^l\circ\cn_{s^l}\circ\cP^l)d\mu_l(\cP^l,\crr^l)\ge 1-2^{-la}
\end{equation}
for some positive constant $a\in\rr$.\\
Let $\eps\in (0,1)$. Then for all sufficiently large $l\in\nn$ there exist $l^2$ codes $\{(\cP^l_i,\crr^l_i):i=1,\ldots ,l^2\}\subset \mathcal C(\fr_l,\hr^{\otimes l})\times\mathcal C(\kr^{\otimes l},\fr'_l)$ such that
\begin{equation}\label{eq:random-code-reduction-a}
\frac{1}{l^2}\sum_{i=1}^{l^2} F_e(\pi_{\fr_l},\crr^l_i\circ\cn_{s^l}\circ\cP^l_i)>1-\eps \qquad  \forall s^n\in\mathbf S^n.
\end{equation}
\end{lemma}
\begin{proof}We define random variables $(\Lambda_i,\Omega_i)$, $i=1,\ldots, l^2$ with values in $\mathcal C(\fr_l,\hr^{\otimes l})\times\mathcal C(\kr^{\otimes l},\fr'_l) $ which are i.i.d. according to $\mu_l^{\otimes l^2}$. Using Markov's inequality and the inequality $2^{\gamma t}\leq (1-t)2^{\gamma\cdot  0}+t2^{\gamma}\le 1+t2^{\gamma},\ t\in[0,1], \gamma>0$ as well as the union bound we get\\
$\mathbb P\left(\frac{1}{l^2}\sum_{i=1}^{l^2} F_e(\pi_{\fr_l},\Lambda_i\circ\cn_{s^l}\circ\Omega_i)>1-\eps\ \forall s^l\in\bS^l\right)\geq1-|\mathbf S|^l\cdot 2^{-l^2\eps}$.
For large enough $l$ the above probability is positive. This shows the existence of the required realization of $(\Lambda_i,\Omega_i)_{i=1}^{l^2}$. \end{proof}
\begin{proof}\emph{(Of the second claim in Theorem \ref{quant-ahlswede-dichotomy})}. As shown above, in order to achieve $\A_r(\fri)$ we need only random codes with discrete support on subexponentially many points. Whenever $C_d(\fri)>0$ and $\A_r(\fri)>0$ the sender can transmit classical information at rate zero over the AVQC in order to derandomize the code without any asymptotic reduction in the capacity for entanglement transmission.
\end{proof}
\subsection{\label{subsec:infinite-avqc}Infinite AVQC's}
Let $\fri=\{\cn_s  \}_{s\in\bS}$ with $|\bS|=\infty$. We consider the set $\tilde\fri:=\overline{conv(\fri)}^{||\cdot||_{\lozenge}}$ - the closure of $conv(\fri)$ w.r.t. $||\cdot ||_{\lozenge}$. Suppose that
\begin{equation}\label{eq:infinite-1}
  D_{\lozenge}(\tilde\fri,\partial\mathcal{C})=:a>0.
\end{equation}
Our goal is to find an outer approximation of $\tilde\fri$ in Hausdorff metric (cf. Section \ref{Notation and Conventions}) by polytopes contained entirely in the set $\mathcal{C}(\hr,\kr)$. To this end, we need the following result of convex analysis (cf. Theorem 3.1.6, p. 109, in \cite{webster}).
\begin{theorem}\label{polytope-approximation}
Let $A$ be a non-empty compact convex set in $\rr^d$ and let $\eps>0$. Then there exist polytopes $P,Q$ in $\rr^d$ such that $P\subseteq A\subseteq Q$ and $D(A,P)\le \eps$, $D(A,Q)\le \eps$, where $D(\cdot,\cdot)$ denotes the Hausdorff distance induced by the euclidean norm on $\rr^d$.
\end{theorem}
We note that the presence of $\rr^d$ and the euclidean norm in Theorem \ref{polytope-approximation} is not essential at all. The theorem holds as well for any finite dimensional normed space with corresponding Hausdorff distance induced by the given norm.\\
\begin{proof}\emph{(Of Theorem \ref{avqc-infinite}.)} We apply Theorem \ref{polytope-approximation} to the space $H(\hr,\kr):=\mathcal{B}_{h}(\mathcal{B}(\hr),\mathcal{B}(\kr))$ of hermiticity preserving linear maps from $\mathcal{B}(\hr)$ into $\mathcal{B}(\kr)$ endowed with $||\cdot||_{\lozenge}$ and obtain for each $\eps>0$ a polytope $\bar{Q}_{\eps}$ with $\tilde\fri \subseteq \bar{Q}_{\eps}\quad \textrm{and}\quad D_{\lozenge}( \tilde\fri,\bar{Q}_{\eps} )\le\eps.$\\
Let $E$ denote the affine hull of $\mathcal{C}(\hr,\kr)$ in $H(\hr,\kr)$ and set $Q_{\eps}:=E\cap \bar{Q}_{\eps}$. Then $Q_{\eps}$ is a polytope and for all sufficiently small $\eps>0$ ($\eps\le \frac{a}{3}$, say, is small enough for this purpose) we have $\tilde\fri \subseteq Q_{\eps}\subset \mathcal{C}(\hr,\kr)$ by (\ref{eq:infinite-1}). More important, we also have
\begin{equation}\label{eq:infinite-2}
  D_{\lozenge}(\tilde\fri, Q_{\eps})\le D_{\lozenge}(\tilde\fri, \bar{Q}_{\eps})\le \eps.
\end{equation}
Let $\fri_{\eps}=\{ \cn_1,\ldots, \cn_K \}$ be the extremal points of $Q_\eps$. Then $\fri_\eps$ has the following properties:
1) $conv(\fri)\subset\tilde\fri\subset Q_\eps=conv(\fri_\eps)$, 2) $D_{\lozenge}(\tilde\fri, conv(\fri_{\eps}))\le \eps $ for all sufficiently small $\eps>0$ by (\ref{eq:infinite-2}).\\
We can now apply all results from Section \ref{subsec:random-achieve-finite} to the finite AVQC generated by $\fri_{\eps}$ giving us to each sufficiently small $\eta>0$ and $k\in\nn$ a sequence of $(l,k_l)$-random codes $(\cP^l,\crr^l)_{l\in\nn}$ with $\cP^l\in\mathcal{C}(\fr_l,\hr^{\otimes l})$, $\crr^l\in\mathcal{C}(\kr^{\otimes l},\fr'_l)$,
\begin{equation}\label{eq:infinite-3}
  F_e(\pi_{\fr_l},\crr^l\circ \cn_{t^l}\circ\cP^l)\ge 1-(l+1)^{K}\cdot 2^{-lc} \qquad \forall \ t^l\in \{1,\ldots, K  \}^{l},
\end{equation}
and
\begin{equation}\label{eq:infinite-4}
\frac{1}{l}\log k_l\ge \frac{1}{k}\inf_{\cn\in conv(\fri_{\eps})}I_c(\rho,\cn^{\otimes k})-\frac{\eta}{2},
\end{equation}
for any $\rho\in\cs (\hr^{\otimes k})$ and all sufficiently large $l\in\nn$ with a positive constant $c=c(k,\dim\hr,\dim\kr,\fri_{\eps},\eta)$.
Since $\fri\subseteq\tilde\fri\subseteq conv(\fri_{\eps})$ we can find to any finite collection $\cn'_1,\ldots,\cn'_l\in\fri$ probability distributions $q_1,\ldots ,q_l\in\mathfrak{P}(\{1,\ldots,K  \})$ with $\cn'_i=\sum_{j=1}^{K}q_i(j)\cn_j\qquad (\cn_j\in\fri_{\eps}, j\in\{1,\ldots,K   \}).$
Thus, for any choice of $\cn'_1,\ldots, \cn'_l\in\fri$
\begin{eqnarray}\label{eq:infinite-5}
  F_e(\pi_{\fr_l},\crr^l\circ(\otimes_{i=1}^l\cn'_i )\circ \cP^l)\ge 1-(l+1)^{K}\cdot 2^{-lc},
\end{eqnarray}
by (\ref{eq:infinite-3}).
On the other hand, Lemma 16 in  \cite{bbn-2} and $D_{\lozenge}(conv(\fri), conv(\fri_{\eps}))\le D_{\lozenge}(\tilde\fri, conv(\fri_{\eps}))\le \eps$ shows that
\begin{equation}\label{eq:infinite-7}
 \frac{1}{l}\log k_l\ge \frac{1}{k}\inf_{\cn\in conv(\fri)}I_c(\rho,\cn^{\otimes k})-{\eta},
\end{equation}
whenever $\eps$ is small enough. It should be noted that $k$ and $l$ in the above equation tend to infinity when $\eta$ goes to zero. Since $\eta>0$ was arbitrary, we are done.
\end{proof}
\section{\label{sec:symmetrizability}Symmetrizability}
In this section we introduce a notion of symmetrizability which is a sufficient and necessary condition for $C_{\textrm{det}}(\fri)=0$. Our approach is motivated by the corresponding concept for arbitrarily varying channels with classical input and quantum output (cq-AVC) given in \cite{ahlswede-blinovsky}. In what follows we will restrict ourselves to the case $|\bS|<\infty$.
\begin{definition}\label{def:c-symmetrizability}
Let $\bS$ be a finite set and $\fri=\{\cn_s  \}_{s\in\bS}$ an AVQC.
\begin{enumerate}
\item $\fri$ is called $l$-symmetrizable, $l\in\nn$, if for each finite set $\{\rho_1,\ldots,\rho_K  \}\subset \cs(\hr^{\otimes l})$, $K\in \nn$, there is a map $p:\{\rho_1,\ldots, \rho_K  \}\to \mathfrak{P}(\bS^l)$ such that for all $i,j\in\{1,\ldots, K  \}$ the following holds:
\begin{equation}\label{eq:c-symmetrizable}
\sum_{s^l\in\bS^l}p(\rho_i)(s^l)\cn_{s^l}(\rho_j)= \sum_{s^l\in\bS^l}p(\rho_j)(s^l)\cn_{s^l}(\rho_i).
\end{equation}
\item We call $\fri$ symmetrizable if it is $l$-symmetrizable for all $l\in\nn$.
\end{enumerate}
\end{definition}
\begin{theorem}\label{symm-equiv-C-0}
Let $\fri=\{\cn_s  \}_{s\in \bS}$, $|\bS|<\infty$, be an AVQC. Then $\fri$ is symmetrizable if and only if $C_{\textup{det}}(\fri)=0$.
\end{theorem}
\begin{proof}The proof follows closely the corresponding arguments given in \cite{ericson}, \cite{csiszar-narayan}, and \cite{ahlswede-blinovsky}. \end{proof}
\begin{corollary}\label{q-det-=0}
If the AVQC $\fri=\{\cn  \}_{s\in \bS}$ is symmetrizable then $\mathcal{A}_{\textup{d}}(\fri)=0$.
\end{corollary}
\begin{proof}Note that $\mathcal{A}_{\textrm{d}}(\fri)\le C_{\textrm{det}}(\fri)  $ and apply Theorem \ref{symm-equiv-C-0}.
\end{proof}
What is missing now is the reverse direction in Corollary \ref{q-det-=0}: That an AVQC with $\mathcal{A}_{\textup{d}}(\fri)=0$ is symmetrizable. It is not known yet whether this implication is true or not.\\
The final issue in this section is a sufficient condition for $\mathcal{A}_{\textup{r}}(\fri)=0$ which is based on the notion of qc-symmetrizability. Let $\B_+(\hr)\subset\B(\hr)$ be the set of nonnegative operators. Set
\begin{eqnarray*}\textrm{QC}(\hr,\bS):=\{\{T_s\}_{s\in\bS}\subset\B_+(\hr):\sum_{s\in\bS}T_s=\mathbf 1_\hr\}.\end{eqnarray*}
For a given finite set of quantum channels $\fri=\{\cn_s  \}_{s\in \bS}$ and $T\in \textrm{QC}(\hr,\bS) $ we define a CPTP map $\mathcal{M}_{T,\bS}:\mathcal{B}(\hr)\otimes \mathcal{B}(\hr)\to\mathcal{B}(\kr)$ by
\begin{eqnarray}\label{def-M}
\mathcal{M}_{T,\bS}(a\otimes b):=\sum_{s\in \bS}\textrm{tr}(T_s a)\cn_s(b).
\end{eqnarray}
\begin{definition}\label{def:symmetrizability}
An arbitrarily varying quantum channel, generated by a finite set $\fri=\{\cn_s  \}_{s\in \bS}$, is called qc-symmetrizable if there is $T\in \textrm{QC}(\hr,\bS) $ such that for all $a,b\in\mathcal{B}(\hr) $
\begin{equation}\label{eq:symmetrizable}
  \mathcal{M}_{T,\bS}(a\otimes b)=\mathcal{M}_{T,\bS}(b\otimes a)
\end{equation}
holds, where $\mathcal{M}_{T,\bS}:\mathcal{B}(\hr)\otimes \mathcal{B}(\hr)\to\mathcal{B}(\kr) $ is the CPTP map defined in (\ref{def-M}).
\end{definition}
The best illustration of the definition of qc-symmetrizability is given in the proof of our next theorem:
\begin{theorem}\label{symm-implies-0-capacity}
If an arbitrarily varying quantum channel generated by a finite set $\fri=\{\cn_s  \}_{s\in \bS}$ is qc-symmetrizable, then for any sequence of $(l,k_l)$-random codes $(\mu_l)_{l\in\nn}$ with $k_l=\dim \fr_l\ge 2$ for all $l\in\nn$ we have
\[\inf_{s^l\in \bS^l}\int F_e(\pi_{\fr_l}, \crr^l\circ \cn_{s^l}\circ \mathcal{P}^l )d\mu_l (\crr^l,\mathcal{P}^l)\le \frac{1}{2}, \]
for all $l\in\nn$. Thus $\mathcal{A}_{\textup{r}}(\fri)=0$, and consequently
\[\mathcal{A}_{\textup{d}}(\fri)=0.  \]
\end{theorem}
\emph{Remark: Our Definition \ref{def:symmetrizability} addresses the notion of qc-symmetrizability for block length $l=1$. In our accompanying paper \cite{abbn-1} we show that the corresponding definition for arbitrary $l$ is equivalent.}
\begin{proof}Let $l\in\nn$. We fix $\sigma\in \cs(\hr)$ and define $E_1,E_2 \in\mathcal C(\hr,\kr)$ by $E_1(a):=\mathcal{M}_{T,\bS}(\sigma\otimes a)$,
\begin{eqnarray}\label{eq:0-cap-3}
E_2(a):=\mathcal{M}_{T,\bS}(a\otimes \sigma)=\sum_{s\in \bS}\textrm{tr}(E_s a)\cn_s(\sigma).
\end{eqnarray}
Setting $E_{s^l}:= E_{s_1}\otimes \ldots \otimes E_{s_l}$, we can show that
\begin{eqnarray}\label{eq:0-cap-5}
  &&\int F_e(\pi_{\fr_l}, \crr^l\circ E_1^{\otimes l}\circ \mathcal{P}^l)d\mu_l (\crr^l,\mathcal{P}^l )\geq\nonumber\\
&&\inf_{s^l\in \bS^l}\int F_e(\pi_{\fr_l}, \crr^l\circ \cn_{s^l}\circ \mathcal{P}^l )d\mu_l (\crr^l,\mathcal{P}^l ).
\end{eqnarray}
Now, by the assumed qc-symmetrizability, we get
\begin{eqnarray}
 id_{\fr_l}\otimes (\crr^l\circ E_1^{\otimes l} )=id_{\fr_l}\otimes (\crr^l\circ E_2^{\otimes l} ),\ \textrm{thus}
\end{eqnarray}
\begin{equation}\label{eq:0-cap-7}
F_e(\pi_{\fr_l}, \crr^l\circ E_1^{\otimes l}\circ \mathcal{P}^l)=F_e(\pi_{\fr_l}, \crr^l\circ E_2^{\otimes l}\circ \mathcal{P}^l).
\end{equation}
But $E_2$ is entanglement breaking, implying that $(id_{\fr_l}\otimes \crr^l\circ E_2^{\otimes l}\circ \mathcal{P}^l)(|\psi_l\rangle\langle \psi_l|)$ (for a purification $\psi_l$ of $\pi_{\fr_l}$)
is separable. A standard result from entanglement theory implies that
\begin{equation}\label{eq:0-cap-9}
\langle \psi_l, (id_{\fr_l}\otimes \crr^l\circ E_2^{\otimes l}\circ \mathcal{P}^l)(|\psi_l\rangle\langle \psi_l|)        \psi_l\rangle \le \frac{1}{k_l}
\end{equation}
holds, since $\psi_l$ is maximally entangled with Schmidt rank $k_l$. Combining (\ref{eq:0-cap-5}), (\ref{eq:0-cap-7}), (\ref{eq:0-cap-9}) and our assumption $k_l\geq2$ we get\\
$\inf_{s^l\in \bS^l}\int F_e(\pi_{\fr_l}, \crr^l\circ \cn_{s^l}\circ \mathcal{P}^l )d\mu_l(\crr^l,\mathcal{P}^l)\le \frac{1}{2}$.
\end{proof}


\end{document}